\documentclass[10pt,conference]{IEEEtran}
\usepackage{epsf,psfrag,amssymb,amsfonts,cite,amsmath,eufrak,yfonts}

\newtheorem{theorem}{Theorem}
\newtheorem{example}{Example}

\newtheorem{definition}{Definition}

\def\psfancypar#1#2{\begingroup\def\par{\endgraf\endgroup\lineskiplimit=0pt}
               \setbox2=\hbox{\large\sc #2}
               \newdimen\tmpht \tmpht \ht2 \advance\tmpht by \baselineskip
               \font\hhuge=Times-Bold at \tmpht
               \setbox1=\hbox{{\hhuge #1}}
               \count7=\tmpht \count8=\ht1
               \divide\count8 by 1000 \divide\count7 by \count8 
               \tmpht=.001\tmpht\multiply\tmpht by \count7 
               \font\hhuge=Times-Bold at \tmpht
               \setbox1=\hbox{{\hhuge #1}}
               \noindent
                \hangindent1.05\wd1
               \hangafter=-2 {\hskip-\hangindent
               \lower1\ht1\hbox{\raise1.0\ht2\copy1}%
                \kern-0\wd1}\copy2\lineskiplimit=-1000pt}

\newcommand{\beq}{\begin{equation}}
\newcommand{\eeq}{\end{equation}}
\newcommand{\bqa}{\begin{eqnarray}}
\newcommand{\eqa}{\end{eqnarray}}
\newcommand{\bqn}{\begin{eqnarray*}}
\newcommand{\eqn}{\end{eqnarray*}}
\newcommand{\nn}{\nonumber}

\newcommand{\be}{\begin{enumerate}}
\newcommand{\ee}{\end{enumerate}}
\newcommand{\bi}{\begin{itemize}}
\newcommand{\ei}{\end{itemize}}
\newcommand{\bd}{\begin{description}}
\newcommand{\ed}{\end{description}}
\newcommand{\ba}{\begin{array}}
\newcommand{\ea}{\end{array}}
\newcommand{\bde}{\begin{definition}}
\newcommand{\ede}{\end{definition}}
\newcommand{\bex}{\begin{example}}
\newcommand{\eex}{\end{example}}

\def\boxit#1{\vbox{\hrule\hbox{\vrule\kern3pt
        \vbox{\kern3pt#1\kern3pt}\kern3pt\vrule}\hrule}}

\def\reals{ { {\rm  I \kern-0.15em R }  } }
\def\complex{ {\,{{\rm C} \kern-0.50em \raise0.20ex {  |}}\, }}

\def\0bf{{\bf 0}}
\def\1bf{{\bf 1}}
\def\2bf{{\bf 2}}
\def\3bf{{\bf 3}}
\def\4bf{{\bf 4}}
\def\5bf{{\bf 5}}
\def\6bf{{\bf 6}}
\def\7bf{{\bf 7}}
\def\8bf{{\bf 8}}
\def\9bf{{\bf 9}}

\def\Ibf{{\bf I}}

\def\Rbf{{\bf R}}

\def\Emat{\mathcal{E}}

\def\Rxx{\Rbf_{\ssstyle X\kern-.1em X}}

\let\ssstyle=\scriptscriptstyle

\def\Kout{\setbox1=\hbox{\Huge\bf K}\hbox to
1.05\wd1{\hspace{.05\wd1}
}}
\def\Sout{\setbox1=\hbox{\Huge\bf S}\hbox to 1.05\wd1{\hspace{.05\wd1}
}}

\allowdisplaybreaks[1]
\begin{document}

\title{New Outer Bounds on the Capacity Region of Gaussian Interference Channels }

\author{\authorblockN{Xiaohu Shang}
\authorblockA{Syracuse University\\
Department of EECS\\
Email: xshang@syr.edu} \and
\authorblockN{Gerhard Kramer}
\authorblockA{Bell Labs\\
Alcatel-Lucent\\
Email: gkr@research.bell-labs.com} \and
\authorblockN{Biao Chen}
\authorblockA{Syracuse University\\
Department of EECS\\
Email: bichen@ecs.syr.edu}}\maketitle

\begin{abstract}
Recent outer bounds on the capacity region of Gaussian
interference channels are generalized to $m$-user channels with
$m>2$ and asymmetric powers and crosstalk coefficients. The bounds
are again shown to give the sum-rate capacity for Gaussian
interference channels with low powers and crosstalk coefficients.
The capacity is achieved by using single-user detection at each
receiver, i.e., treating the interference as noise incurs no loss
in performance.

\end{abstract}
 \maketitle

\section{Introduction}

This paper extends the results of \cite{Shang-etal:08CISS} to
asymmetric Gaussian ICs. The paper further has a new Theorem
(Theorem \ref{thm:mIC}) that is not in
\cite{Shang-etal:06IT_submission} or in other recent works (see
Section V and
 Motahari and Khandani \cite{Motahari&Khandani:08IT_submission}, and Annapureddy and Veeravalli \cite{Annapureddy&Veeravalli:08ITA}).

The interference channel (IC) models communication systems where
transmitters communicate with their respective receivers while
causing interference to all other receivers. For a two-user
Gaussian IC, the channel output can be written in the standard
form \cite{Carleial:78IT} \bqn
Y_1&=&X_1+\sqrt{a}X_2+Z_1,\\Y_2&=&\sqrt{b}X_1+X_2+Z_2,\eqn where
$\sqrt{a}$ and $\sqrt{b}$ are channel coefficients, $X_i$ and
$Y_i$ are the transmit and receive signals. The user/channel input
sequence $X_{i1},X_{i2},\cdots,X_{in}$ is subject to the power
constraint $\sum_{j=1}^n\Emat(X^2_{ij})\leq nP_i$, $i=1,2$. The
transmitted signals $X_1$ and $X_2$ are statistically independent.
The channel noises $Z_1$ and $Z_2$ are possibly correlated unit
variance Gaussian random variables, and $(Z_1,Z_2)$ is
statistically independent of $(X_1,X_2)$. In the following, we
denote this Gaussian IC as IC$(a,b,P_1,P_2)$.

The capacity region of an IC is defined as the closure of the set
of rate pairs $(R_1,R_2)$ for which both receivers can decode
their own messages with arbitrarily small positive error
probability. The capacity region of a Gaussian IC is known only
for three cases: (1) $a=0$, $b=0$. (2) $a\geq 1$, $b\geq 1$: see
\cite{Carleial:75IT,Sato:81IT,Han&Kobayashi:81IT}. (3) $a=0$,
$b\geq 1$; or $a\geq 1$, $b=0$: see \cite{Costa:85IT}. For the
second case both receivers can decode the messages of both
transmitters. Thus this IC acts as two multiple access channels
(MACs), and the capacity region for the IC is the intersection of
the capacity region of the two MACs. However, when the
interference is weak or moderate, the capacity region is still
unknown. The best inner bound is obtained in
\cite{Han&Kobayashi:81IT} by using superposition coding and joint
decoding. A simplified form of the Han-Kobayashi region was given
by Chong-Motani-Garg-El Gamal \cite{Chong-etal:06IT_submission},
\cite{Kramer:06Zurich}. Various outer bounds have been developed
in
\cite{Sato:77IT,Carleial:83IT,Kramer:04IT,Etkin-etal:07IT_submission,Telatar&Tse:07ISIT}.
Kramer in \cite{Kramer:04IT} presented two outer bounds. The first
is obtained by providing each receiver with just enough
information to decode both messages. The second is obtained by
reducing the IC to a degraded broadcast channel. Both bounds
dominate the bounds by Sato \cite{Sato:77IT} and Carleial
\cite{Carleial:83IT}. The recent outer bounds by Etkin, Tse, and
Wang in \cite{Etkin-etal:07IT_submission} are also based on
genie-aided methods, and they show that Han and Kobayashi's inner
bound is within one bit or a factor of two of the capacity region.
This result can also be established by the methods of Telatar and
Tse \cite{Telatar&Tse:07ISIT}. We remark that neither of the
bounds of \cite{Kramer:04IT} and \cite{Etkin-etal:07IT_submission}
implies each other. Numerical results show that the bounds of
\cite{Kramer:04IT} are better at low SNR while those of
\cite{Etkin-etal:07IT_submission} are better at high SNR. The
bounds of \cite{Telatar&Tse:07ISIT} are not amenable to numerical
evaluation since the optimal distributions of the auxiliary random
variables are unknown.

In this paper,  we present new outer bounds on the capacity region
of Gaussian ICs that generalize results of
\cite{Shang-etal:08CISS}.  The new bound is based on a genie-aided
approach and an extremal inequality proposed by Liu and Viswanath
\cite{Liu&Viswanath:06IT}. Based on this outer bound, we obtain
new sum-rate capacity results
 for ICs satisfying some channel
coefficient and power constraint conditions. We show that the
sum-rate capacity can be achieved by treating the interference as
noise when both the channel gain and the power constraint are
weak. We say that such channels have {\em noisy interference}. For
this class of interference, the simple single-user transmission
and detection strategy is sum-rate optimal.

This paper is organized as follows. In Section II, we present an
outer bound and the resulting sum-rate capacity for certain
$2$-user Gaussian ICs. An extension of the sum-rate capacity under
noisy interference to $m$-user ICs is provided in Section III.
Numerical examples are given in Section IV, and Section V
concludes the paper.

\section{A Genie-aided Outer Bound}
\subsection{General outer bound}
The following is a new outer bound on the capacity region of
Gaussian ICs. Note that in contrast to \cite{Shang-etal:08CISS}
these bounds permit $P_1\neq P_2$ and $a\neq b$.
\begin{theorem}
If the rates $(R_1,R_2)$ are achievable for IC$(a,b,P_1,P_2)$ with
$0<a<1,0<b<1$, they must satisfy the following constraints
(\ref{eq:constraint1})-(\ref{eq:constraint3}) for $\mu>0$,
$\frac{1+bP_1}{b+bP_1}\leq\eta_1\leq\frac{1}{b}$ and
$a\leq\eta_2\leq\frac{a+aP_2}{1+aP_2}$, where
\begin{figure*}
\bqa R_1+\mu
R_2&\leq&\min_{\substack{\rho_i\in[0,1]\\\left(\sigma_1^2,\sigma_2^2\right)\in\Sigma}}\frac{1}{2}\log\left(1+\frac{P_1^*}{\sigma_1^2}\right)-\frac{1}{2}\log\left(aP_2^*+1-\rho_1^2\right)+\frac{1}{2}\log\left(1+P_1+aP_2-\frac{(P_1+\rho_1\sigma_1)^2}{P_1+\sigma_1^2}\right)\nn\\
&&\hspace{.2in}+\frac{\mu}{2}\log\left(1+\frac{P_2^*}{\sigma_2^2}\right)-\frac{\mu}{2}\log\left(bP_1^*+1-\rho_2^2\right)+\frac{\mu}{2}\log\left(1+P_2+bP_1-\frac{(P_2+\rho_2\sigma_2)^2}{P_2+\sigma_2^2}\right)\label{eq:constraint1}\\
R_1+\eta_1R_2&\leq&\frac{1}{2}\log\left(1+\frac{b\eta_1-1}{b-b\eta_1}\right)-\frac{\eta_1}{2}\log\left(1+\frac{b\eta_1-1}{1-\eta_1}\right)+\frac{\eta_1}{2}\log\left(1+bP_1+P_2\right)\label{eq:constraint2}\\
R_1+\eta_2R_2&\leq&\frac{1}{2}\log\left(1+P_1+aP_2\right)-\frac{1}{2}\log\left(1+\frac{a-\eta_2}{\eta_2-1}\right)+\frac{\eta_2}{2}\log\left(1+\frac{a-\eta_2}{a\eta_2-a}\right).\label{eq:constraint3}\eqa
 \noindent \rule{\textwidth}{.02cm}
\vspace{-.3in}
\end{figure*}

 \bqa\Sigma=\left\{\begin{array}{ll}
\left\{\left(\sigma_1^2,\sigma_2^2\right)\left.|\sigma_1^2>0,0<\sigma_2^2\leq\frac{1-\rho_1^2}{a}\right.\right\},\textrm{ if }\mu\geq 1, \\
\left\{\left(\sigma_1^2,\sigma_2^2\right)\left.| 0<\sigma_1^2\leq\frac{1-\rho_2^2}{b},\sigma_2^2>0\right.\right\},\textrm{ if }\mu<1, \\
 \end{array}\right.\eqa and if $\mu\geq 1$ we define
\bqa
&&\hspace{-.3in}P_1^*=\left\{\begin{array}{ll}P_1,&\sigma_1^2\leq\left[\frac{(1-\mu)P_1}{\mu}+\frac{1-\rho_2^2}{b\mu}\right]^+,\\
\frac{1-\rho_2^2-b\mu\sigma_1^2}{b\mu-b},&
\left[\frac{(1-\mu)P_1}{\mu}+\frac{1-\rho_2^2}{b\mu}\right]^+<\sigma_1^2\leq\frac{1-\rho_2^2}{b\mu},\\
0,&\sigma_1^2>\frac{1-\rho_2^2}{b\mu},\end{array}\right.\label{eq:P1muL}\\
&&\hspace{-.3in}P_2^*=P_2,\qquad\qquad\qquad
\sigma_2^2\leq\frac{1-\rho_1^2}{a}, \label{eq:P2muL}\eqa where
$(x)^+\triangleq\max\{x,0\}$, and if $0<\mu<1$ we define\bqa
&&\hspace{-.3in}P_1^*=P_1,\qquad\qquad\qquad\quad
\sigma_1^2\leq\frac{1-\rho_2^2}{b},\label{eq:P1muS}\\
&&\hspace{-.3in}P_2^*=\left\{\begin{array}{ll}P_2,&\sigma_2^2\leq\left[\left(\mu-1\right)P_2+\frac{\mu\left(1-\rho_1^2\right)}{a}\right]^+,\\
\frac{\mu\left(1-\rho_1^2\right)-a\sigma_2^2}{a-a\mu},&
\hspace{-.1in}\begin{array}{l}
  \left[\left(\mu-1\right)P_2+\frac{\mu\left(1-\rho_1^2\right)}{a}\right]^+ \\
  \hspace{.8in}<\sigma_2^2\leq\frac{\mu\left(1-\rho_1^2\right)}{a}, \\
\end{array}\\
0,&\sigma_2^2>\frac{\mu\left(1-\rho_1^2\right)}{a}.\end{array}\right.\label{eq:P2muS}
 \eqa\label{thm:region}
\end{theorem}

\begin{proof} We give a sketch of the proof. A genie provides the two receivers with side information $X_1+N_1$ and $X_2+N_2$
respectively, where $N_i$ is Gaussian distributed with variance
$\sigma_i^2$ and $\Emat(N_iZ_i)=\rho_i\sigma_i$, $i=1,2$. Starting
from Fano's inequality, we have\bqa &&n(R_1+\mu
R_2)-n\epsilon\nn\\
&&\leq I\left(X_1^n;Y_1^n\right)+\mu
I\left(X_2^n;Y_2^n\right)\nn\\
&&\leq I\left(X_1^n;Y_1^n,X_1^n+N_1^n\right)+\mu
I\left(X_2^n;Y_2^n,X_2^n+N_2^n\right)\nn\\
&&=\left[h\left(X_1^n+N_1^n\right)-\mu
h\left(\sqrt{b}X_1^n+Z_2^n|N_2^n\right)\right]-h\left(N_1^n\right)\nn\\
&&\hspace{.2in}+\left[\mu
h\left(X_2^n+N_2^n\right)-h\left(\sqrt{a}X_2^n+Z_1^n|N_1^n\right)\right]-\mu h\left(N_2^n\right)\nn\\
&&\hspace{.2in}+h\left(Y_1^n|X_1^n+N_1^n\right)+\mu
h\left(Y_2^n|X_2^n+N_2^n\right),\label{eq:weightedsum}\eqa where
$\epsilon\rightarrow 0$ as $n\rightarrow\infty$. For
$h\left(Y_1^n|X_1^n+N_1^n\right)$, zero-mean Gaussian $X_1^n$ and
$X_2^n$  with covariance matrices $P_1\Ibf$ and $P_2\Ibf$ are
optimal, and we have \bqa
&&\hspace{-.3in}h\left(Y_1^n|X_1^n+N_1^n\right)\nn\\
&&\hspace{-.3in}\leq\frac{n}{2}\log\left[2\pi
e\left(1+aP_2+P_1-\frac{\left(P_1+\rho_1\sigma_1\right)^2}{P_1+\sigma_1^2}\right)\right].\label{eq:eq1}\eqa
From the extremal inequality introduced in \cite[Theorem 1,
Corollary 4]{Liu&Viswanath:06IT}, we have \bqa
&&\hspace{-.2in}h\left(X_1^n+N_1^n\right)-\mu
h\left(\sqrt{b}X_1^n+Z_2^n|N_2^n\right)\label{eq:firstCombine}\\
&&\hspace{-.2in}\leq\frac{n}{2}\log\left[2\pi
e\left(P_1^*+\sigma_1^2\right)\right]-\frac{n\mu}{2}\log\left[2\pi
e\left(bP_1^*+1-\rho_2^2\right)\right],\nn\eqa and \bqa
&&\hspace{-.2in}\mu
h\left(X_2^n+N_2^n\right)-h\left(\sqrt{a}X_2^n+Z_1^n|N_1^n\right)\label{eq:secondCombine}\\
&&\hspace{-.2in}\leq\frac{n\mu}{2}\log\left[2\pi
e\left(P_2^*+\sigma_2^2\right)\right]-\frac{\mu}{2}\log\left[2\pi
e\left(aP_2^*+1-\rho_1^2\right)\right],\nn\eqa where equalities
hold when $X_1^n$ and $X_2^n$ are both Gaussian with covariance
matrices $P_1^*\Ibf$ and $P_2^*\Ibf$ respectively. From
(\ref{eq:weightedsum})-(\ref{eq:secondCombine}) we obtain the
outer bound (\ref{eq:constraint1}).

The outer bound in (\ref{eq:constraint2}) (resp.
(\ref{eq:constraint3})) is obtained by letting the genie provide
side information $X_2$ to receiver one (resp. $X_1$ to receiver
two), and applying the extremely inequality, i.e., \bqa
&&n(R_1+\eta_1R_2)-n\epsilon\nn\\
&&\leq I\left(X_1^n;Y_1^n,X_2^n\right)+\eta_1I\left(X_2^n;Y_2^n\right)\nn\\
&&=h\left(X_1^n+Z_1^n\right)-\eta_1h\left(\sqrt{b}X_1^n+Z_2^n\right)-h\left(Z_1^n\right)\nn\\
&&\hspace{.15in}+\eta_1h\left(Y_2^n\right)\nn\\
&&\leq\frac{n}{2}\log\left(\tilde P_1+1\right)-\frac{n\eta_1}{2}\log\left(b\tilde P_1+1\right)\nn\\
&&\hspace{.15in}+\frac{n\eta_1}{2}\log\left(1+bP_1+P_2\right),\label{eq:BCbound}\eqa
where $\tilde P_1=\frac{b\eta_1-1}{b-b\eta_1}$ for
$\frac{1+bP_1}{b+bP_1}\leq\eta_1\leq\frac{1}{b}$. This is the
bound in (\ref{eq:constraint2}). Similarly, we obtain bound
(\ref{eq:constraint3}).
\end{proof}

 Remark 1: The bounds
(\ref{eq:constraint1})-(\ref{eq:constraint3}) are obtained by
providing different genie-aided signals to the receivers. There is
overlap of the range of $\mu$, $\eta_1$, and $\eta_2$, and none of
the bounds uniformly dominates the other two bounds.

Remark 2: Equations (\ref{eq:constraint2}) and
(\ref{eq:constraint3}) are outer bounds for the capacity region of
a Z-IC, and a Z-IC is equivalent to a degraded IC
\cite{Costa:85IT}. For such channels, it can be shown that
(\ref{eq:constraint2}) and (\ref{eq:constraint3}) are the same as
the outer bounds in \cite{Sato:78IT}. For
$\eta_1=\frac{1+bP_1}{b+bP_1}$ and $\eta_2=\frac{a+aP_2}{1+aP_2}$,
the bounds in (\ref{eq:constraint2}) and (\ref{eq:constraint3})
are tight for a Z-IC (or degraded IC) because $\tilde
P_1=P_1,\tilde P_2=P_2$ in (\ref{eq:BCbound}), and there is no
power sharing between the transmitters. Consequently,
$\frac{1+bP_1}{b+bP_1}$ and $\frac{a+aP_2}{1+aP_2}$ are the
negative slopes of the tangent lines for the capacity region at
the corner points.

Remark 3: The bounds in
(\ref{eq:constraint2})-(\ref{eq:constraint3}) turn out to be the
same as the bounds in \cite[Theorem 2]{Kramer:04IT}. This can be
shown by rewriting the bounds in \cite[Theorem 2]{Kramer:04IT} in
the form of a weighted sum rate.

Remark 4: The bounds in \cite[Theorem 2]{Kramer:04IT} are obtained
by getting rid of one of the interference links to reduce the IC
into a Z-IC. In addition, the proof in \cite{Kramer:04IT} allowed
the transmitters to share their power, which further reduces the
Z-IC into a degraded broadcast channel. Then the capacity region
of this degraded broadcast channel is an outer bound for the
capacity region of the original IC. The bounds in
(\ref{eq:constraint2}) and (\ref{eq:constraint3}) are also
obtained by reducing the IC to a Z-IC. Although we do not
explicitly allow the transmitters to share their power, it is
interesting that these bounds are equivalent to the bounds in
\cite[Theorem 2]{Kramer:04IT} with power sharing. In fact, a
careful examination of our derivation reveals that power sharing
is implicitly assumed. For example, for the term
$h\left(X_1^n+Z_1^n\right)-\eta_1h\left(\sqrt{b}X_1^n+Z_2^n\right)$
of (\ref{eq:BCbound}), user $1$ uses power $\tilde
P_1=\frac{b\eta_1-1}{b-b\eta_1}\leq P_1$ , while for the term
$\eta_1h\left(Y_2^n\right)$ user $1$ uses all the power $P_1$.
This is equivalent to letting user $1$ use the power $\tilde P_1$
for both terms, and letting user $2$ use a power that exceeds
$P_2$.

Remark 5: Theorem \ref{thm:region} improves \cite[Theorem
3]{Etkin-etal:07IT_submission}. Specifically, the bound in
(\ref{eq:constraint2}) is tighter than the first sum-rate bound of
\cite[Theorem 3]{Etkin-etal:07IT_submission}. Similarly, the bound
in (\ref{eq:constraint3}) is tighter than the second sum-rate
bound of \cite[Theorem 3]{Etkin-etal:07IT_submission}. The third
sum-rate bound in \cite[Theorem 3]{Etkin-etal:07IT_submission} is
a special case of (\ref{eq:constraint1}).

Remark 6: Our outer bound is not always tighter than that of
\cite{Etkin-etal:07IT_submission} for all rate points. The reason
is that in \cite[last two equations of
(39)]{Etkin-etal:07IT_submission}, different genie-aided signals
are provided to the same receiver. Our outer bound can also be
improved in a similar and more general way by providing different
genie-aided signals to the receivers. Specifically the starting
point of the bound can be modified to be \bqa n\left(R_1+\mu
R_2\right)&\leq&\sum_{i=1}^k\lambda_iI\left(X_1^n;Y_1^n,U_i\right)\nn\\&&+\sum_{j=1}^m\mu_iI\left(X_2^n;Y_2^n,W_j\right)+n\epsilon,\label{eq:extension}\eqa
where
$\sum_{i=1}^k\lambda_i=1,\sum_{j=1}^m\mu_j=\mu,\lambda_i>0,\mu_j>0$.

\subsection{Sum-rate capacity for noisy interference}

The outer bound in Theorem \ref{thm:region} is in the form of an
optimization problem. Four parameters
$\rho_1,\rho_2,\sigma_1^2,\sigma_2^2$ need to be optimized for
different choices of the weights $\mu,\eta_1,\eta_2$. When
$\mu=1$, the bound (\ref{eq:constraint1}) of Theorem
\ref{thm:region} leads directly to the following sum-rate capacity
result.

\begin{theorem}
For the IC$(a,b,P_1,P_2)$ satisfying \bqa
\sqrt{a}(bP_1+1)+\sqrt{b}(aP_2+1)\leq 1,\label{eq:power}\eqa
 the sum-rate capacity is \bqa
C=\frac{1}{2}\log\left(1+\frac{P_1}{1+aP_2}\right)+\frac{1}{2}\log\left(1+\frac{P_2}{1+bP_1}\right).\label{eq:sumcapacity}\eqa
\label{thm:sumcapacity}\end{theorem}

\begin{proof}By choosing \bqa
&&\hspace{-.25in}\sigma_1^2=\frac{1}{2b}\left\{b(aP_2+1)^2-a(bP_1+1)^2+1\right.\nn\\
&&\left.\pm\sqrt{\left[b(aP_2+1)^2-a(bP_1+1)^2+1\right]^2-4b(aP_2+1)^2}\right\}\label{eq:condition1}\nn\\
&&\hspace{-.25in}\sigma_2^2=\frac{1}{2a}\left\{a(bP_1+1)^2-b(aP_2+1)^2+1\right.\nn\\
&&\left.\pm\sqrt{\left[a(bP_1+1)^2-b(aP_2+1)^2+1\right]^2-4a(bP_1\emph{}+1)^2}\right\}\nn\\
&&\hspace{-.25in}\rho_1=\sqrt{1-a\sigma_2^2}\label{eq:condition3}\\
&&\hspace{-.25in}\rho_2=\sqrt{1-b\sigma_1^2},
\label{eq:condition4}\eqa the bound (\ref{eq:constraint1}) with
$\mu=1$ is \bqa
R_1+R_2\leq\frac{1}{2}\log\left(1+\frac{P_1}{1+aP_2}\right)+\frac{1}{2}\log\left(1+\frac{P_2}{1+bP_1}\right).\label{eq:upperlower}\eqa
But one can achieve equality in (\ref{eq:upperlower}) by treating
the interference as noise at both receivers. In order that the
choices of $\sigma_i^2$ and $\rho_i^2$ are feasible,
(\ref{eq:power}) must be satisfied.
\end{proof}

Remark 7: Consider the bound (\ref{eq:constraint1}) with $\mu=1$,
we further let \bqa 1-\rho_1^2\geq a\sigma_2^2,\quad
1-\rho_2^2\geq b\sigma_1^2.\label{eq:conditionalPower}\eqa From
(\ref{eq:P1muL}) and (\ref{eq:P2muL}) we have
$P_1^*=P_1,P_2^*=P_2$. Thus,
 \bqa &&\hspace{-.2in}R_1\leq\frac{1}{2}\log\left(1+\frac{P_1}{\sigma_1^2}\right)-\frac{1}{2}\log(aP_2+1-\rho_1^2)\nn\\
&&\hspace{.2in}+\frac{1}{2}\log\left[1+aP_2+P_1-\frac{(P_1+\rho_1\sigma_1)^2}{P_1+\sigma_1^2}\right]\nn\\
&&=\frac{1}{2}\log\left[\frac{P_1(1+aP_2)}{1+aP_2-\rho_1^2}\left(\frac{1}{\sigma_1}-\frac{\rho_1}{1+aP_2}\right)^2+1+\frac{P_1}{1+aP_2}\right]\nn\\
&&\triangleq f(\rho_1,\sigma_1).\label{eq:symmsum}\eqa Therefore,
for any given $\rho_1$, when
\bqa\rho_1\sigma_1=1+aP_2,\label{eq:sigma}\eqa then
$f(\rho_1,\sigma_1)$ achieves its minimum which is user $1$'s
single-user detection rate. Similarly, we have
$\rho_2\sigma_2=1+bP_1$. Since the constraint in
(\ref{eq:conditionalPower}) must be satisfied, we have \bqa
\frac{1+aP_2}{\rho_1}\leq\sqrt{\frac{1-\rho_2^2}{b}},\quad
\frac{1+bP_1}{\rho_2}\leq\sqrt{\frac{1-\rho_1^2}{a}}.\label{eq:rhosigma}\eqa
As long as there exists a $\rho_i\in(0,1)$ such that
(\ref{eq:rhosigma}) is satisfied, we can choose $\sigma_i$ to
satisfy (\ref{eq:sigma}) and hence the bound in
(\ref{eq:constraint1}) is tight. By cancelling $\rho_1,\rho_2$, we
obtain (\ref{eq:power}).

Remark 8: The most special choices of $\rho_1,\rho_2$ are in
(\ref{eq:condition3}) and (\ref{eq:condition4}),  since
(\ref{eq:firstCombine}) and (\ref{eq:secondCombine}) with $\mu=1$
become \bqn h\left(X_1^n+N_1^n\right)-
h\left(\sqrt{b}X_1^n+Z_2^n|N_2^n\right)&=&-n\log\sqrt{b}\nn\\
h\left(X_2^n+N_2^n\right)-h\left(\sqrt{a}X_2^n+Z_1^n|N_1^n\right)&=&-n\log\sqrt{a}.\eqn
Therefore, we do not need the extremal inequality
\cite{Liu&Viswanath:06IT} to prove Theorem \ref{thm:sumcapacity}.

Remark 9: The sum-rate capacity for a Z-IC with $a=0$, $0<b<1$ is
a special case of Theorem \ref{thm:sumcapacity} since
(\ref{eq:power}) is satisfied. The sum-rate capacity is therefore
given by (\ref{eq:sumcapacity}).

Remark 10: Theorem \ref{thm:sumcapacity} follows directly from
Theorem \ref{thm:region} with $\mu=1$. It is remarkable that a
genie-aided bound is tight if (\ref{eq:power}) is satisfied since
the genie provides extra signals to the receivers without
increasing the rates. This situation is reminiscent of the recent
capacity results for vector Gaussian broadcast channels (see
\cite{Weingarten-etal:06IT}). Furthermore, the sum-rate capacity
(\ref{eq:sumcapacity}) is achieved by treating the interference as
noise. We therefore refer to channels satisfying (\ref{eq:power})
as ICs with {\em noisy interference}.

Remark 11: For a symmetric IC where $a=b,P_1=P_2=P$, the
constraint (\ref{eq:power}) implies that \bqa a\leq
\frac{1}{4},\quad
P\leq\frac{\sqrt{a}-2a}{2a^2}.\label{eq:channelgain}\eqa {\em
Noisy interference} is therefore {\em weaker} than {\em weak
interference} as defined in \cite{Costa:85IT} and
\cite{Sason:04IT}, namely $a\leq\frac{\sqrt{1+2P}-1}{2P}$ or \bqa
a\leq \frac{1}{2},\quad P\leq\frac{1-2a}{2a^2}.\label{eq:weak}\eqa
Recall that \cite{Sason:04IT} showed that for ``weak interference"
satisfying (\ref{eq:weak}), treating interference as noise
achieves larger sum rate than time- or frequency-division
multiplexing (TDM/FDM), and \cite{Costa:85IT} claimed that in
``weak interference" the largest known achievable sum rate is
achieved by treating the interference as noise.

\subsection{Capacity region corner point}
The bounds (\ref{eq:constraint2}) and (\ref{eq:constraint3}) of
Theorem \ref{thm:region} lead to the following sum-rate capacity
result.
\begin{theorem}
For an IC$(a,b,P_1,P_2)$ with $a>1$, $0<b<1$, the sum-rate
capacity is \bqa
C=\frac{1}{2}\log\left(1+P_1\right)+\frac{1}{2}\log\left(1+\frac{P_2}{1+bP_1}\right)\label{eq:sumcapacity_mixed}\eqa
when the following condition holds \bqa (1-ab)P_1\leq
a-1.\label{eq:constraint_mixed}\eqa A similar result follows by
swapping $a$ and $b$, and $P_1$ and $P_2$.
\label{thm:sumcapacity_mixed}\end{theorem}

This sum-rate capacity is achieved by a simple scheme: user $1$
transmits at the maximum rate and user $2$ transmits at the rate
that both receivers can decode its message with single-user
detection. Such a rate constraint was considered in \cite[Theorem
1]{Costa:85IT} which established a corner point of the capacity
region. However it was pointed out in \cite{Sason:04IT} that the
proof in \cite{Costa:85IT} was flawed. Theorem
\ref{thm:sumcapacity_mixed} shows that the rate pair of
\cite{Sason:04IT} is in fact a corner point of the capacity region
when $a>1, 0<b<1$ and (\ref{eq:constraint_mixed}) is satisfied,
and this rate pair achieves the sum-rate capacity.

 The sum-rate
capacity of the degraded IC $(ab=1,0<b<1)$ is a special case of
Theorem \ref{thm:sumcapacity_mixed}. Besides this example, there
are two other kinds of ICs to which Theorem
\ref{thm:sumcapacity_mixed} applies. The first case is $ab>1$. In
this case, $P_1$ can be any positive value. The second case is
$ab<1$ and $P_1\leq\frac{a-1}{1-ab}$. For both cases, the signals
from user $2$ can be decoded first at both receivers.

\section{Sum-rate Capacity for $m$-user IC with Noisy Interference}
For an $m$-user IC, the receive signal at user $i$ is defined as
\bqa Y_i=X_i+\sum_{j=1,j\neq
i}^m\left(\sqrt{c_{ji}}X_j\right)+Z_i,\quad
i=1,2,\dots,m,\label{eq:mIC}\eqa where $c_{ji}$ is the channel
gain from $j^{th}$ transmitter to $i^{th}$ receiver, $Z_i$ is
unit-variance Gaussian noise, and the transmit signals have the
block power constraints $\sum_{l=1}^n\Emat(X_{il}^2)\leq nP_i$. We
have the following sum-rate capacity result.
\begin{theorem}
For an $m$-user IC defined in (\ref{eq:mIC}), if there exist
$\rho_i\in(0,1),i=1,\dots,m$, such that the following conditions
are satisfied \bqa \sum_{j=1,j\neq i}^m\frac{c_{ji}(1+Q_j)^2}{\rho_j^2}&\leq&1-\rho_i^2\label{eq:mICcondition1}\\
 \sum_{j=1,j\neq i}^m\frac{c_{ij}}{1+Q_j-\rho_j^2}&\leq&\frac{1}{P_i+(1+Q_i)^2/\rho_i^2}\label{eq:mICcondition2},\eqa
where $Q_i$ is the interference power at receiver $i$, defined as
\bqa Q_i=\sum_{j=1,j\neq i}^mc_{ji}P_j,\eqa the sum-rate capacity
is \bqa
C=\frac{1}{2}\sum_{i=1}^m\log\left(1+\frac{P_i}{1+Q_i}\right)\label{eq:mICsumcapacity}\eqa\label{thm:mIC}
\end{theorem}
Therefore, if there exist $\rho_1,\dots,\rho_m$, such that
(\ref{eq:mICcondition1}) and (\ref{eq:mICcondition2}) are
satisfied for all $i=1,\dots,m$, the sum-rate capacity of an
$m$-user IC can be achieved by treating interference as noise. The
proof is omitted due to the space limitation. It can be shown that
Theorem \ref{thm:sumcapacity} is a special case of Theorem
\ref{thm:mIC}.

Consider a uniformly symmetric $m$-user IC where $c_{ji}=c$, for
all $i,j=1,\dots,m, i\neq j$, and $P_i=P$. The bounds
(\ref{eq:mICcondition1}) and (\ref{eq:mICcondition2}) with
$\rho_i=\rho$ for all $i$ reduce to \bqa
c\leq\frac{1}{4(m-1)},\quad P\leq
\frac{\sqrt{(m-1)c}-2(m-1)c}{2(m-1)^2c^2}.\label{eq:symmNoisyICPower}\eqa

\section{Numerical examples}

A comparison of the outer bounds for a Gaussian IC is given in
Fig. \ref{fig:region}. Some part of the outer bound from Theorem
\ref{thm:region} overlaps with Kramer's outer bound due to
(\ref{eq:constraint2}) and (\ref{eq:constraint3}). Since this IC
has noisy interference, the proposed outer bound coincides with
the inner bound at the sum rate point.

The lower and upper bounds for the sum-rate capacity of the
symmetric IC are shown in Fig \ref{fig:sumVeryStrong} for high
power level. The upper bound is tight up to point $A$. The bound
in \cite[Theorem 3]{Etkin-etal:07IT_submission} approaches the
bound in Theorem \ref{thm:region} when the power is large, but
there is still a gap. Fig. \ref{fig:sumVeryStrong} also provides a
definitive answer to a question from \cite[Fig. 2]{Sason:04IT}:
whether the sum-rate capacity of symmetric Gaussian IC is a
decreasing function of $a$, or there exists a bump like the lower
bound when $a$ varies from $0$ to $1$. In Fig.
\ref{fig:sumVeryStrong}, our proposed upper bound and Sason's
inner bound explicitly show that the sum capacity is not a
monotone function of $a$ (this result also follows by the bounds
of \cite{Etkin-etal:07IT_submission}).

\section{Conclusions, extensions and parallel work}

We derived an outer bound for the capacity region of $2$-user
Gaussian ICs by a genie-aided method. From this outer bound, the
sum-rate capacities for ICs that satisfy (\ref{eq:power}) or
(\ref{eq:constraint_mixed}) are obtained. The sum-rate capacity
for $m$-user Gaussian ICs  that satisfy (\ref{eq:mICcondition1})
and (\ref{eq:mICcondition2}) are also obtained.

Finally, we wish to acknowledge parallel work. After submitting
our $2$-user bounds and capacity results on the arXiv.org e-Print
archive \cite{Shang-etal:06IT_submission}, two other teams of
researchers - Motahari and Khandani from the University of
Waterloo, Annapureddy and Veeravalli from the University of
Illinois at Urbana-Champaign - let us know that they derived the
same $2$-user sum-rate capacity results (Theorem
\ref{thm:sumcapacity}).

\vspace{.2in}

 \centerline{Acknowledgement} The work of X. Shang and B. Chen
was supported in part by the NSF under Grants 0546491 and 0501534,
and by the AFOSR under Grant FA9550-06-1-0051, and by the AFRL
under Agreement FA8750-05-2-0121.

G. Kramer gratefully acknowledges the support of the Board of
Trustees of the University of Illinois Subaward no. 04-217 under
NSF Grant CCR-0325673 and the Army Research Office under ARO Grant
W911NF-06-1-0182.

\begin{figure}[htp]
\centerline{\leavevmode \epsfxsize=3.3in \epsfysize=2.2in
\epsfbox{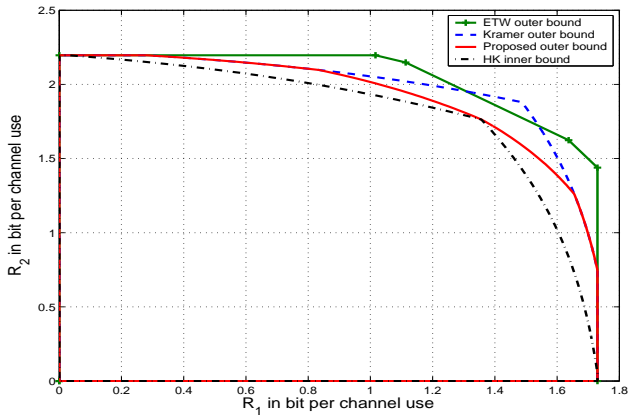}}\caption{Inner and outer bounds for the
capacity region of Gaussian ICs with
$a=0.04,b=0.09,P_1=10,P_2=20$. The ETW bound is by Etkin, Tse and
Wang in \protect\cite[Theorem 3]{Etkin-etal:07IT_submission}; the
Kramer bound is from \protect\cite[Theorem 2]{Kramer:04IT}; the HK
inner bound is based on \protect\cite{Han&Kobayashi:81IT} by Han
and Kobayashi.} \label{fig:region}\end{figure}

\begin{figure}[htp]
\centerline{\leavevmode \epsfxsize=3.5in \epsfysize=2.4in
\epsfbox{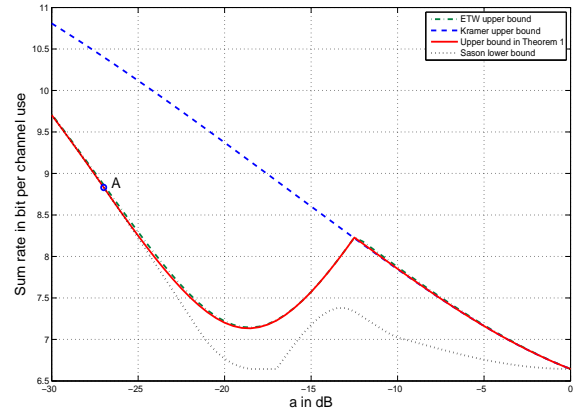}}\caption{Lower and upper bounds for the
sum-rate capacity of symmetric Gaussian ICs with
$a=b,P_1=P_2=5000$. The channel gain at point $A$ is
$a=-26.99$dB. Sason's bound is an inner bound obtained from Han
and Kobayashi's bound by a special time sharing scheme
\protect\cite[Table I]{Sason:04IT}.}
\label{fig:sumVeryStrong}\end{figure}

\bibliography{Journal,Conf}
\bibliographystyle{IEEEbib}

\end{document}